\documentclass[reprint,superscriptaddress,amsmath,amssymb,aps,pra]{revtex4-2}
\usepackage[colorlinks, linkcolor = ForestGreen, anchorcolor = blue, citecolor = magenta]{hyperref}
\usepackage{amsmath,amssymb,amsthm,mathtools,mathrsfs}
\usepackage[table,dvipsnames]{xcolor}
\usepackage{tikz}
\usepackage{adjustbox}
\usepackage{dsfont}
\usepackage[normalem]{ulem}
\usepackage{enumerate}
\usepackage[all]{xy} 
\usepackage{physics}

\theoremstyle{plain}
\newtheorem{theorem}{Theorem}
\newtheorem{lemma}{Lemma}
\newtheorem{proposition}{Proposition}

\theoremstyle{definition}
\newtheorem{definition}{Definition}

\newtheorem{example}{Example}

\newcommand{\bd}{\begin{definition}}
\newcommand{\ed}{\end{definition}}
\newcommand{\bt}{\begin{theorem}}
\newcommand{\et}{\end{theorem}}
\newcommand{\bn}{\begin{proposition}}
\newcommand{\en}{\end{proposition}}
\newcommand{\be}{\begin{equation}}
\newcommand{\ee}{\end{equation}}
\newcommand{\blem}{\begin{lemma}}
\newcommand{\elem}{\end{lemma}}
\newcommand{\bx}{\begin{example}}
\newcommand{\ex}{\end{example}}
\newcommand{\bprf}{\begin{proof}}
\newcommand{\eprf}{\end{proof}}

\DeclareMathAlphabet{\mathpzc}{OT1}{pzc}{m}{it} 
 \DeclareFontFamily{OT1}{pzc}{}
 \DeclareFontShape{OT1}{pzc}{m}{it}{ <-> s*[1.2] pzcmi7t }{}
 \DeclareMathAlphabet{\mathpzc}{OT1}{pzc}{m}{it}
 
\newcommand{\id}{\operatorname{id}}

\def\E{\mathcal{E}}
\def\H{\mathcal{H}}

\providecommand{\Tr}{\operatorname{Tr}}

\providecommand{\E}{\mathcal{E}}
\providecommand{\HU}{\mathbb{E}_{U}}

\newtheorem*{theoremone}{Theorem 1}
\newtheorem*{theoremtwo}{Theorem 2}
\newtheorem*{theoremthree}{Theorem 3}
\newtheorem*{theoremfour}{Theorem 4}

\begin{document}									
\preprint{APS/123-QED}

\title{Quantum information loss}
\author{James Fullwood}
\email{fullwood@hainanu.edu.cn}
\affiliation{School of Mathematics and Statistics, Hainan University, Haikou, Hainan, 570228, China}
\affiliation{Hainan International Exchange Center for Theoretical Physics, Haikou, Hainan, 570228, China}

\author{Wu-zhong Guo}
\email{wuzhong@hust.edu.cn}
\affiliation{School of Physics, Huazhong University of Science and Technology, Wuhan, Hubei, 430074, China}
\affiliation{Center for Gravitational Physics and Quantum Information,\\
 Yukawa Institute for Theoretical Physics, Kyoto University, Kyoto 606-8502, Japan}

\author{Boyu Yang}
\affiliation{School of Mathematics and Statistics, Hainan University, Haikou, Hainan, 570228, China}

\date{\today}

\begin{abstract}
We introduce a measure of information loss for any quantum process that may be modeled by a prepare-evolve-measure scenario: Alice prepares an ensemble of states that gets sent via a quantum channel to Bob, who then measures the output. As a quantum channel models open system dynamics, our measure of information loss quantifies Bob's inability to retrodict with certainty which state Alice sent through the channel. By minimizing this measure over all possible pure state ensemble decompositions of a fixed state $\rho$, and over all POVMs on the output of a channel $\E$, we arrive at an intrinsic notion of information loss for any state-channel pair $(\rho,\E)$. We show that the vanishing of information loss with respect to all states supported on a fixed codespace $\H_{\text{code}}$ is equivalent to a condition we term \emph{universal pristineness}, which ensures that orthogonal pure states in $\H_{\text{code}}$ get sent via the channel $\E$ to possibly mixed states whose supports are orthogonal. Moreover, we prove universal pristineness is equivalent to the Knill-Laflamme conditions in quantum error correction, which are necessary and sufficient for the existence of a perfect recovery channel for all states supported on $\H_{\text{code}}$. As an application, we apply our framework to the Hayden-Preskill model of black hole evaporation, demonstrating that the evaporation channel becomes asymptotically universally pristine, thereby providing a purely channel-theoretic formulation of Page-time information retrieval.
\end{abstract}

\maketitle

\section{Introduction}

The concept of information loss lies at the heart of various aspects of fundamental physics. From quantum decoherence to Landauer's principle and the black hole information problem, tracking the flow of information in quantum processes has solidified the tenet that ``information is physical''~\cite{Zeh_1970,Joos_1985,Landauer_1961,Zurek_2003,Busch_2009,Hawking_1976, Akil_2025, Landauer_1996}. However, a general consensus has yet to be established regarding a rigorous formulation of information loss associated with quantum processes and how to quantify it. 

While it is often stipulated that the information loss associated with a quantum process is quantified by the difference between the von~Neumann entropies of its initial and final states, such a formulation of information loss tells us nothing about the associated dynamics. In particular, an initial and final state of a quantum process may have the same entropy whether it be unitary evolution or dissipative dynamics corresponding to system-environment interactions. Therefore, the difference of entropies between the initial and final states of a quantum process is fundamentally insufficient to capture the information loss induced by a specified dynamics. 

Conversely, recent work by Li \textit{et al.} \cite{Li_2025} utilizes the quadratic entropy of the Choi state of a quantum channel to quantify information loss associated with open system dynamics. While this state-independent formulation provides a powerful characterization of the channel's intrinsic dynamics, it leaves room to explore the operational aspects of information loss by accounting for the preparation of an initial state to be sent through a channel. While Cerf defines a loss function in Ref.~\cite{Cerf_1997} that explicitly incorporates both an initial state and a quantum channel governing the dynamics of a quantum process, this approach relies on the introduction of an auxiliary reference system. Building on this foundation, there is a natural opportunity to develop a mathematically precise, reference-free formulation of information loss that integrates both initial states and dynamics within an inherently operational framework.

In this work, we adapt the axiomatic approach of Refs.~\cite{Baez_2011,FuPa21} on classical information loss to the operational setting of prepare-evolve-measure (PEM) scenarios associated with quantum systems: Alice prepares a quantum state $\rho_i$ with probability $p_i$, and then sends the state via a quantum channel $\E$ to Bob who subsequently performs a measurement on the channel's output. We then define the information loss of such a PEM scenario as the conditional entropy of Alice's preparation given Bob's measurement outcome. Operationally, we show that the information loss associated with a PEM scenario quantifies Bob's inability to retrodict with certainty which state Alice sent through the channel. As such a measure of information loss depends on the uncertainty of Alice's preparation, the choice of Bob's measurement, and the underlying dynamics associated with $\E$, it incorporates all relevant operational data. Moreover, we show that an intrinsic, measurement-independent notion of quantum information loss is obtained by minimizing this quantity over all pure state ensemble decompositions of a fixed state $\rho$ and over all POVMs on the output of the quantum channel $\E$.

We also prove a necessary and sufficient condition for the vanishing of the intrinsic information loss associated with a state-channel pair $(\rho,\E)$, which we refer to as the \emph{pristine} condition. We show that unitary channels are pristine with respect to any initial state $\rho$, and more generally, that the pristine condition is equivalent to the existence of a perfect recovery channel for all states on a fixed code space, establishing a precise operational interpretation of information loss in terms of quantum error correction~\cite{Knill_2000}. Moreover, we show that the intrinsic information loss of the pair $(\rho,\E)$ is bounded above by the von~Neumann entropy $S(\rho)$, which is achieved for example by the completely depolarizing channel. 

As an application of our framework, we consider a PEM scenario in the context of black hole evaporation, where Alice throws a fixed number of qubits into a black hole and Bob then measures the Hawking radiation. Under the Hayden-Preskill model~\cite{Hayden_2007}—where random unitaries drive the internal scrambling of information—the evolution mapping Alice's qubits to the total radiation is modeled as a quantum channel. Within this setup, we show that the associated process becomes nearly pristine as the size of the newly emitted radiation exceeds the number of qubits Alice threw into the black hole, thus providing a channel-theoretic formulation of Page-time information retrieval~\cite{Page_1993}.

\section{Information loss in PEM scenarios}
 Let $A$ and $B$ denote quantum systems with finite-dimensional Hilbert spaces $\H_A$ and $\H_B$. The algebra of linear operators on $\H_X$ with $X\in \{A,B\}$ will be denoted by $\mathcal{L}(\H_X)$, and states of $X$ will be represented by density operators on $\H_X$. Throughout this work, we suppose that Alice prepares system $A$ in state $\rho_i$ with probability $p_i$, and then sends the state to Bob via a quantum channel $\E:\mathcal{L}(\H_A)\to \mathcal{L}(\H_B)$, which is a completely positive, trace-preserving linear map which models system-environment interactions. We further assume that Bob will measure the output of the channel $\E$ with a POVM $\{N_j\}$ on system $B$, so that $\{N_j\}$ is a collection of positive operators on $\H_B$ which sum to the identity operator $\mathds{1}_B$. Such a setup will be referred to as a \emph{prepare-evolve-measure} (PEM) scenario, and will be denoted by the triple $(\{(p_i,\rho_i)\},\E,\{N_j\})$. To avoid measure-zero subtleties we assume $p_i>0$ for all $i$. 

Given a PEM scenario $(\{(p_i,\rho_i)\},\E,\{N_j\})$, the joint probability $P(i,j)$ that Alice prepares the state $\rho_i$ and Bob's obtains the measurement outcome $N_j$ is given by $P(i,j)=p_i\Tr\big(\E(\rho_i) N_j\big)$. It then follows that the marginal distribution $Q(j)=\sum_i P(i,j)$ corresponding to the probability that Bob obtains the outcome $N_j$---regardless of which state Alice prepares---is given by $Q(j)=\Tr(\E(\rho)N_j)$, where $\rho=\sum_ip_i\rho_i$. We then define the \emph{information loss} $K$ associated with the PEM scenario $(\{(p_i,\rho_i)\},\E,\{N_j\})$ to be the conditional entropy of Alice's preparation \emph{given} Bob's measurement outcome, which is given by $K=H(P)-H(Q)$, where $H(\cdot)$ denotes the Shannon entropy. We note that it follows from standard properties of Shannon entropy that $0\leq K\leq H(p)$, where $p$ is the distribution associated with Alice's preparation.

If we minimize $K$ with respect to all possible pure state ensemble decompositions $(p_i,\dyad{\psi_i}{\psi_i})$ of a fixed state $\rho$ of $A$, and all possible POVMs on system $B$, we obtain the \emph{intrinsic information loss} associated with the pair $(\rho,\E)$, which will be denoted by $K(\rho,\E)$. We note that by taking a spectral decomposition of $\rho$, the upper bound $K\leq H(p)$ implies that $K(\rho,\E)\leq S(\rho)$, where $S(\cdot)$ denotes the von~Neumann entropy.

The following result is fundamental.
\bt \label{PTX71}
Let $\{(p_i,\rho_i)\}$ be an ensemble of states of $A$, and let $\E:\mathcal{L}(\H_A)\to \mathcal{L}(\H_B)$ be a quantum channel. Then there exists a POVM $\{N_j\}$ on system $B$ such that $K=0$ if and only if $\E(\rho_i)\E(\rho_k)=0$ for all $i\neq k$. 
\et

In light of Theorem~\ref{PTX71} (which will be proved in the Supplemental Material, along with the proofs of Theorems~\ref{ECXN57}-\ref{thm:HP_asymptotic_universal_pristine}), a channel $\E$ will be referred to as \emph{pristine} with respect to an ensemble $\{(p_i,\rho_i)\}$ of states of $A$ if and only if $\E(\rho_i)\E(\rho_k)=0$ for all $i\neq k$. If $\E$ is pristine with respect to an ensemble of pure states $\{(p_i,\dyad{\psi_i}{\psi_i})\}$, let $\Pi_i$ be the orthogonal projector onto the image of $\E(\dyad{\psi_i}{\psi_i})$ for all $i$, and define $N_i=\Pi_i$ and $N_{\mathrm{rest}}=\mathds{1}_B-\sum_{i}\Pi_i$. It then follows from the proof of Theorem~\ref{PTX71} that $\{N_j\}=\{N_i\}\cup\{N_{\mathrm{rest}}\}$ is a projective measurement such that $K=0$. Moreover, as the intrinsic information loss $K(\rho,\E)$ is obtained by minimizing $K$ over all pure state ensemble decompositions of $\rho$ and over all POVMs on $B$, it follows that $K(\rho,\E)=0$ whenever $\E$ is pristine with respect to an ensemble of pure states $\{(p_i,\dyad{\psi_i}{\psi_i})\}$ such that $\rho=\sum_ip_i \dyad{\psi_i}{\psi_i}$.

If $\E$ is a unitary channel, so that $\E(\rho)=U\rho U^{\dag}$ for some unitary operator on $\H_A$, then it follows from Theorem~\ref{PTX71} that $K(\rho,\E)=0$ for any state $\rho$ of $A$. Indeed, if we consider the pure state ensemble decomposition of $\rho$ corresponding to its spectral decomposition $\rho=\sum_i p_i\dyad{i}{i}$, then for all $i\neq k$ we have 
\begin{align*}
\E(\dyad{i}{i})\E(\dyad{k}{k})&=U\dyad{i}{i}U^{\dag}U\dyad{k}{k}U^{\dag} \\
&=U\braket{i}{k}\dyad{i}{k}U^{\dag}=0\, ,
\end{align*}
thus $\E$ is pristine with respect to the spectral decomposition of $\rho$. Therefore, it follows from Theorem~\ref{PTX71} that $K(\rho,\E)=0$ for any state $\rho$ whenever $\E$ is a unitary channel.

At the other extreme, $K=H(p)$ whenever $\E$ is a discard-and-prepare channel, i.e., if there exists a fixed state $\sigma$ of $B$ such that $\E(\omega)=\sigma$ for all states $\omega$ of $A$. To see this, let $\{N_j\}$ be a POVM on $B$. Then in such a case we have $Q(j)=\Tr(\sigma N_j)$ and $P(i,j)=p_i\Tr(\sigma N_j)=p_iQ(j)$ for all $i$ and $j$. By standard properties of Shannon entropy~\cite{Cover_2006}, it follows that $H(P)=H(p)+H(Q)$, thus
\[
K=H(P)-H(Q)=H(p)+H(Q)-H(Q)=H(p)\, .
\]
Therefore, all the information of Alice's preparation $\{(p_i,\rho_i)\}$ is lost in such a case, as expected. As $K=H(p)$ holds independent of the POVM $\{N_j\}$, it follows that the the intrinsic information loss $K(\rho,\E)$ is obtained by minimizing $H(p)$ over all pure state decompositions $\rho=\sum_i p_i \dyad{\psi_i}{\psi_i}$. As the probability distribution corresponding to the eigenvalues of $\rho$ majorize any such distribution $p$~\cite{Nielsen_2001}, it follows that $K(\rho,\E)=S(\rho)$, thus attaining the upper bound $K(\rho,\E)\leq S(\rho)$.

Furthermore, the intrinsic information loss satisfies the data-processing inequality $K(\rho,\E)\leq K(\rho,\mathcal{F}\circ \E)$ for any subsequent quantum channel $\mathcal{F}:\mathcal{L}(\H_B)\to \mathcal{L}(\H_C)$. This follows directly from the Heisenberg picture: for any pure state ensemble decomposition $\{(p_i,\psi_i)\}$ of $\rho$ and any POVM $\{M_j\}$ on system $C$, the complete positivity and unitality of the adjoint map $\mathcal{F}^{\dag}$ ensures that $N_j=\mathcal{F}^{\dag}(M_j)$ constitutes a valid POVM on system $B$. The joint distribution obtained by measuring $\{M_j\}$ after the composite channel $\mathcal{F}\circ \E$ is therefore identical to the distribution obtained by measuring $\{N_j\}$ directly after $\E$. Since $K(\rho,\E)$ is the infimum over all such pure state ensembles and all POVMs on $B$, it is bounded from above by the conditional entropy associated with this specific configuration, yielding the data-processing inequality $K(\rho,\E)\leq K(\rho,\mathcal{F}\circ \E)$.

\section{Information loss and classical retrodiction}
Let $(\{(p_i,\rho_i)\},\E,\{N_j\})$ be a PEM scenario, and let $I$ and $J$ denote the index sets of Alice's preparation and Bob's measurement, respectively. The conditional probabilities $Q(j|i)=\Tr(\E(\rho_i)N_j)$ combine to form a stochastic matrix, which defines a Markov kernel $f:I\to J$. Such a Markov kernel may be viewed as a generalized function whose values are probabilisitic, i.e., $f(i)=j$ with probability $Q(j|i)$. If $Q(j|i)$ is either $0$ or $1$ for all $i$ and $j$, then $f$ may be identified with an actual function, where $f(i)=j$ whenever $Q(j|i)=1$. Given another Markov kernel $g:J\to K$ for some finite set $K$, the composition $g\circ f:I\to K$ is the Markov kernel whose associated stochastic matrix is obtained by multiplying the stochastic matrices associated with $f$ and $g$.

Suppose Alice wants to communicate a message to Bob via the operational setup of the PEM scenario $(\{(p_i,\rho_i)\},\E,\{N_j\})$, so that Alice has a bijective encoding map $e:\mathfrak{A}\to I$ and Bob has a decoding map $d:J\to \mathfrak{A}$, where $\mathfrak{A}$ is the alphabet in which the message will be encoded. Such a communication protocol will be error-free whenever the composition 
\be \label{CMNXP47}
\mathfrak{A}\overset{e}\longrightarrow  I\overset{f}\longrightarrow J\overset{d}\longrightarrow \mathfrak{A}
\ee
is the identity map on the alphabet $\mathfrak{A}$. 

Now suppose there exists a function $g:J\to I$ such that $g\circ f:I\to I$ is the identity map (recall that $f$ is a Markov kernel not a function, so that the composition $g\circ f$ corresponds to matrix multiplication). In such a case, we can take the decoding map $d:J\to \mathfrak{A}$ to be $d=e^{-1}\circ g$, where $e:\mathfrak{A}\to I$ is Alice's encoding map. Plugging this into the composition \eqref{CMNXP47} yields
\begin{align*}
d\circ f\circ e&=(e^{-1}\circ g)\circ f\circ e \\
&=e^{-1}\circ (g\circ f)\circ e \\
&=e^{-1}\circ \id_I\circ\, e \\
&=\id_{\mathfrak{A}}\, ,
\end{align*}
thus the communication protocol will be error-free whenever there exists a function $g:J\to I$ such that $g\circ f$ is the identity map on the index set $I$. In such a case, we refer to the function $g$ as a \emph{perfect retrodiction map}, since given a measurement outcome associated with $\mathcal{N}_j$, Bob can be certain that Alice's input state was $\rho_{g(j)}$. In particular, the set $g^{-1}(i)\subset J$ corresponds to all possible measurement outcomes when Alice inputs state $\rho_i$,  as depicted in Figure~\ref{fig_1}. 

The following result yields a characterization of vanishing information loss in terms of perfect retrodiction maps.

\bt \label{ECXN57}
Let $(\{(p_i,\rho_i)\}_{i\in I},\E,\{N_j\}_{j\in J})$ be a PEM scenario. Then $K=0$ if and only if there exists a perfect retrodiction map $g:J\to I$.
\et

\begin{figure}[t]
    \centering
    \includegraphics[width=\columnwidth]
    {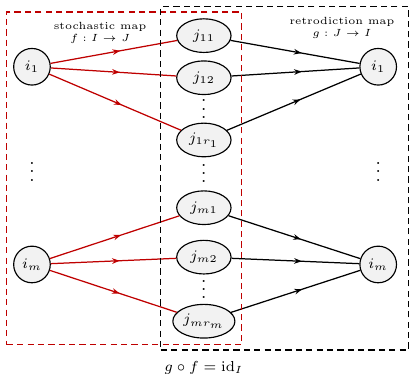}
    \caption{
   When $K=0$, the stochastic map $f:I\to J$ corresponding to the classical labels of Bob's measurement outcomes given Alice's inputs admits a left inverse $g:J\to I$, which we refer to as a perfect retrodiction map.}
    \label{fig_1}
\end{figure}

It follows from Theorem~\ref{ECXN57} that the information loss $K$ quantifies Bob's inability to retrodict with certainty which state Alice sent through the quantum channel $\E$, which may also be viewed as a fundamental degree of error for any communication protocol associated with encoding and decoding maps $e:\mathfrak{A}\to I$ and $d:J\to \mathfrak{A}$. In particular, a positive value $K>0$ strictly precludes the existence of a perfect retrodiction map $g$, since in such a case Bob's measurement outcomes inherently possess overlapping distributions for distinct preparation states, rendering an error-free decoding map $d$ impossible to construct.

To make this fundamental degree of error more precise, we can formally quantify the probability that Bob's classical retrodiction fails. Suppose Bob employs a deterministic guessing strategy to infer Alice's preparation index $i\in I$ from his measurement outcome $j\in J$. To minimize his probability of error, Bob's optimal strategy is to choose the index $i\in I$ that maximizes the joint probability $P(i,j)$. Under this optimal decoding strategy, his minimum probability of error $P_e$ is given by 
\[
P_e=1-\sum_{j\in J}\max_{i\in I}P(i,j) \, .
\]
As the information loss $K$ is defined as the conditional entropy of Alice's preparation given Bob's outcome, Fano's inequality~\cite{Cover_2006} implies that $K\leq f_m(P_e)$, where $m=|I|$ is the number of preparation states and 
\be \label{FMX71}
f_m(x)=h_2(x)+x\log(m-1)\, ,
\ee
where $h_2$ is the binary entropy. Since $f_m(x)$ is strictly monotonically increasing for sufficiently small $x$, a small value of $K$ bounds $P_e$, tightly constraining the probability that the communication protocol fails.

\section{Universal pristineness} 
The pristine condition introduced above is defined relative to a specified ensemble. By Theorems~\ref{PTX71} and \ref{ECXN57}, it characterizes whether Bob can infer the classical label of Alice's preparation without error. We now consider a stronger question, namely, whether one and the same decoding operation can recover an arbitrary input state from a given subspace of $\H_A$, including arbitrary coherent superpositions. This motivates an ensemble-independent strengthening of the pristine condition, which we refer to as `universal pristineness'. For this, let $\H_{\text{code}}\subset \H_A$ be a subspace corresponding to some set of states Alice will send to Bob via a quantum channel $\E$. Density operators of the form $\dyad{\psi}{\psi}$ with $\ket{\psi}\in \H_{\text{code}}$ will be denoted simply by $\psi$. The channel $\E$ is then said to be \emph{universally pristine} with respect to $\H_{\text{code}}$ if
\[
\braket{\psi}{\phi}=0\implies \E(\psi)\E(\phi)=0 \quad \forall \ket{\psi},\ket{\phi}\in \H_{\text{code}}\, .
\] 

The following result establishes an equivalence between universal pristineness, vanishing information loss, and the Knill-Laflamme conditions in quantum error-correction~\cite{Knill_1997}, which are necessary and sufficient for perfect recovery of states in $\H_{\text{code}}$. 

\bt \label{THRMX27}
Suppose $\{E_{\alpha}\}$ is a collection of Kraus operators for the channel $\E$, and let $P$ be the orthogonal projection operator onto $\H_{\text{code}}$. Then the following statements are equivalent.
\begin{enumerate}[i.]
\item \label{THRMX271}
$\E$ is universally pristine with respect to $\H_{\emph{code}}$.
\item \label{THRMX2711}
$K(\rho,\E)=0$ for every state $\rho$ supported on $\H_{\emph{code}}$.
\item \label{THRMX272}
There exists a positive semi-definite
matrix $C=(c_{\alpha\beta})$ such that $P E_{\alpha}^{\dagger}E_{\beta}P=c_{\alpha\beta}P$ for all Kraus indices $\alpha,\beta$.
\item \label{THRMX273}
There exists a quantum channel $\mathcal R:\mathcal L(\H_B)
  \to \mathcal L(\mathcal H_{\emph{code}})$ such that $(\mathcal{R}\circ\E)(\rho)=\rho$ for every state $\rho$
supported on $\H_{\emph{code}}$.
\end{enumerate}
\et
As the equivalence of conditions \ref{THRMX272} and \ref{THRMX273} have been established by the seminal work of Knill and Laflamme on quantum error correction~\cite{Knill_1997}, Theorem~\ref{THRMX27} yields a characterization of quantum error correction in terms of information loss. Specifically, it demonstrates that the capacity to perfectly recover quantum information is fundamentally linked to the vanishing of the intrinsic information loss $K(\rho,\E)$ for all states supported on $\H_{\text{code}}$. This operational perspective allows us to characterize error-correcting codes purely from an information-theoretic standpoint, bypassing the need to explicitly construct a physical recovery channel $\mathcal{R}$. Furthermore, this formulation naturally lends itself to approximate quantum error-correcting codes, where a small but non-zero intrinsic information loss implies that a near-perfect recovery is possible.

To rigorously establish the link to approximate quantum error correction, we can relax the condition of exact universal pristineness by bounding the intrinsic information loss. In particular, if the maximal intrinsic information loss over the codespace is strictly bounded, i.e., $\sup_{\rho} K(\rho,\E) \le \epsilon$ for all states supported on $\H_{\text{code}}$, then the channel's Kraus operators are constrained to approximately satisfy the Knill-Laflamme conditions $P E_{\alpha}^{\dagger} E_{\beta} P \approx c_{\alpha\beta} P$. By standard continuity bounds in approximate error correction~\cite{Beny_2010, Ng_2010}, this constraint guarantees the existence of a completely positive, trace-preserving recovery map $\mathcal{R}$ such that the worst-case root fidelity of recovery, namely,
\[
\min_{\ket{\psi} \in \H_{\text{code}}} F_{\mathrm r}\big(\psi, (\mathcal{R}\circ\E)(\psi)\big)\, , \quad F_{\text{r}}(\rho,\sigma)=\norm{\sqrt{\rho}\sqrt{\sigma}}_1\, ,
\]
is bounded from below by $1 - \mathcal{O}(g(\epsilon))$, where $g(\epsilon)$ is a monotonically increasing function that vanishes as $\epsilon \to 0$. Because perfect recovery corresponds to a root fidelity of exactly $1$, this lower bound strictly constrains the worst-case fidelity to remain arbitrarily close to $1$ for sufficiently small $\epsilon$. As a result, if the intrinsic information loss $K(\rho,\E)$ is small, the recovered state $(\mathcal{R}\circ\E)(\psi)$ is fundamentally indistinguishable from the initial state $\psi$. As such, a small but non-zero intrinsic information loss implies that a near-perfect recovery is always possible without needing to specify the physical mechanics of $\mathcal{R}$---an operational feature that we will exploit in the next section to analyze information retrieval in the context of black hole evaporation.

\section{Information loss in the Hayden-Preskill model} 
The Hayden--Preskill (HP) model for black hole evaporation provides a natural testing ground for our operational framework for information loss, as it describes a concrete process in which quantum information is injected into a black hole, scrambled by its internal dynamics, and subsequently released into the Hawking radiation~\cite{Ha75,Hayden_2007}. In particular, an old black hole past the Page time acts as an ``information mirror'', making newly injected information rapidly recoverable from the collected radiation~\cite{Page_1993}. In the original formulation of Hayden and Preskill, information recovery is diagnosed through the correlations between Alice's input and an auxiliary reference system, while explicit decoding procedures were later constructed by Yoshida and Kitaev ~\cite{Yoshida:2017non}. Here we take a different, reference-free approach: viewing the evaporation process directly as a quantum channel from Alice's input to the collected radiation, we use the intrinsic information loss introduced above to characterize how information becomes accessible in the Hawking radiation.

For this, suppose that Alice throws a $k$-qubits into a black hole that is already maximally entangled with its early radiation. In the HP model, the internal scrambling dynamics of the black hole are modeled by a Haar-random unitary $U$ acting on the combined system of Alice's qubits $A$ (of dimension $d_A=2^k$) and the black hole $H$ (of dimension $d_H=2^n$). Following this scrambling, the black hole emits a subsystem $R'$ consisting of $s$ qubits as new Hawking radiation, leaving behind a remaining black hole system $H'$ of dimension $d_{H'}=2^{k+n-s}$. 

Bob, who has collected both the early radiation $R$ (maximally entangled with $H$ prior to Alice's input) and the newly emitted radiation $R'$, attempts to recover Alice's state. This setup naturally defines a quantum channel $\E_{U,s}:\mathcal{L}(\H_A)\to \mathcal{L}(\H_{RR'})$ given by 
\[
\E_{U,s}(\rho)=\Tr_{H'}\big((U\otimes \mathds{1}_R)(\rho\otimes \dyad{\Psi}{\Psi}_{HR})(U\otimes \mathds{1}_R)^{\dag}\big)\, ,
\]
where $\dyad{\Psi}{\Psi}_{HR}=\frac{1}{\sqrt{d_H}}\sum_{i=1}^{d_H}\ket{i}_H\ket{i}_R$ is the maximally entangled state of $HR$.

\begin{figure}[t]
    \centering
    \includegraphics[width=\columnwidth]
    {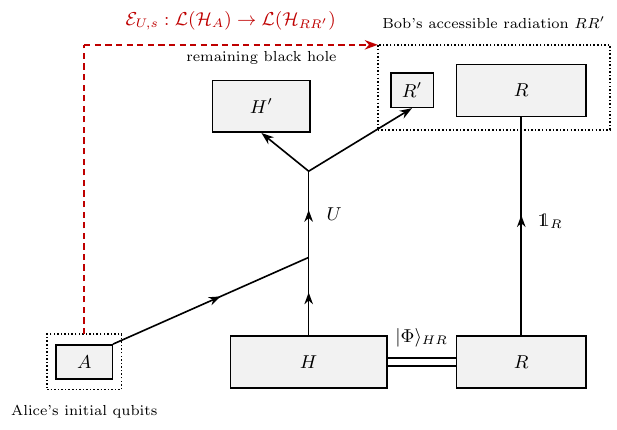}
    \caption{
   Schematic of the Hayden--Preskill model of black hole evaporation, which is modeled as a quantum channel $\E_{U,s}:\mathcal{L}(\mathcal{H}_A)\to \mathcal{L}(\mathcal{H}_{RR'})$. In our formulation, no auxiliary reference system for Alice is required for the recovery of information from the Hawking radiation.}
    \label{fig:HP_channel}
\end{figure}

Our next result establishes that the channel $\E_{U,s}$ approaches universal pristineness on the whole space $\H_A$ as $s$ gets large, which we refer to as \emph{asymptotic universal pristineness}. Moreover, we show how asymptotic universal pristineness of $\E_{U,s}$ implies that the intrinsic information loss $K(\rho_A,\E_{U,s})$ approaches 0 in Haar mean as $s$ gets large. For the statement of the result, let $\rho_A$ be an arbitrary state with spectral decomposition $\rho_A=\sum_{i=1}^{m}p_i\dyad{i}{i}$. We assume that $m>1$ and $p_i>0$ for all $i$, so that $1<m\leq d_A$. Let $C_p=\sum_{i\neq j}\sqrt{p_ip_j}$, and 
\[
\Delta_{\text{HP}}(\rho_A,s)=\min \left\{1-p_{\text{max}},C_p\sqrt{\frac{d_A(d_{H'}^2-1)}{d_{R'}^2d_{H'}^2-1}}\right\}\, ,
\]
where $p_{\text{max}}=\max\{p_i\}_{i=1}^{m}$. We also let $\epsilon(s)$ be the measure of deviation from universal pristineness given by
\[
\epsilon(s)= \sup_{\substack{
        |\psi\rangle,|\phi\rangle\in\mathcal H_A\\
        \langle\psi|\phi\rangle=0
    }}
    \mathbb E_U
    F_{\mathrm r}
    \left(
        \mathcal E_{U,s}(\psi),
        \mathcal E_{U,s}(\phi)
    \right)\, ,
\]
where $F_{\text{r}}(\rho,\sigma)=\norm{\sqrt{\rho}\sqrt{\sigma}}_1$ is the root fidelity and $\mathbb{E}_U$ denotes the Haar mean. 

\bt
\label{thm:HP_asymptotic_universal_pristine}
The following statements hold.
\begin{enumerate}[i.]
\item
$\epsilon(s)\leq 2^{k/2-s}$.
\item
$\Delta_{\emph{HP}}(\rho_A,s)\leq C_p\,2^{k/2-s}$.
\item
$\mathbb{E}_U K(\rho_A,\E_{U,s})\leq f_m\big(\Delta_{\emph{HP}}(\rho_A,s)\big)$, where $f_m$ is the function given by \eqref{FMX71}.
\end{enumerate}
In particular, $\E_{U,s}$ is asymptotically universally pristine on $\H_A$, and the intrinsic information loss $K(\rho_A,\E_{U,s})$ approaches $0$ (in Haar mean) for large $s$.
\et  
As Theorem~\ref{THRMX27} establishes that universal pristineness is equivalent to the existence of a perfect recovery channel for Bob, it follows from Theorem~\ref{thm:HP_asymptotic_universal_pristine} that an unknown quantum state falling into a Hayden-Preskill black hole admits near perfect recovery as the collected radiation grows. In this sense, universal pristineness provides a channel-theoretic characterization of the recoverability underlying the explicit Hayden--Preskill decoding procedures.

At first sight, the bound $\epsilon(s)\leq 2^{k/2-s}$ exhibits a scale $s\sim k/2$, which differs from the usual Hayden--Preskill scale $s\sim k$ obtained from reference-assisted decoupling. This should not be interpreted as a sharper threshold for full information recovery. More precisely, our information-loss bound becomes small when $\Delta_{\mathrm{HP}}(\rho_A,s)\ll1$, and therefore also depends on the input-dependent factor $C_p$. Furthermore, the two approaches quantify different finite-error tasks: our bound controls the Haar-averaged distinguishability of orthogonal inputs and the associated operational information loss, whereas the standard Hayden--Preskill criterion tests the coherent recovery of an arbitrary input, including its correlations with an external reference, by a single decoding channel. Their finite-error scaling need not coincide. In the exact zero-error limit with $s\to +\infty$, however, the distinction disappears, since universal pristineness and perfect quantum recoverability are equivalent by Theorem~\ref{THRMX27}.

\section{Discussion}
In the work, we introduced a rigorous mathematical framework for the formulation of information loss in the operational setting prepare-evolve-measure (PEM) scenarios for quantum systems. By minimizing the information loss over all pure state ensemble decompositions of a fixed state $\rho$, and over all POVMs on the output of a fixed channel $\E$, we arrived at an intrinsic notion of quantum information loss $K(\rho,\E)$ for any state-channel pair $(\rho,\E)$. We proved $0\leq K(\rho,\E)\leq S(\rho)$, with the lower bound being achieved by unitary channels and the upper bound being achieved for completely depolarizing channels. We then showed how the vanishing of information loss with respect to all states supported on a fixed code-space is equivalent to the channel-theoretic condition we termed `universal pristiness', which we also showed is equivalent to the existence of a perfect recovery channel for such states.

Finally, we applied our framework to the Hayden-Preskill model of black hole evaporation, where we were able to show that the associated evaporation channel going from a system of qubits that was thrown into the black hole to the total radiation becomes asymptotically universally pristine as the late radiation grows. While the original work of Hayden-Preskill required a reference system initially entangled with the in-falling qubits for information recovery, our approach is purely channel-theoretic, as we were able to characterize perfect information recovery in terms of the universally pristine condition.

Beyond the specific mechanics of the Hayden-Preskill model, substantial progress has recently been achieved in tractable semiclassical and holographic models of quantum gravity, where quantum extremal surfaces, entanglement islands, and replica wormholes reproduce a Page curve consistent with unitary black hole evaporation \cite{Pe20,Almheiri:2019hni,Almheiri:2019psf,Penington:2019kki,Almheiri:2019qdq}. These developments diagnose information recovery primarily through the fine-grained entropy of the Hawking radiation. Our framework however provides a complementary operational diagnostic that does not rely on computing the Page curve. By treating black hole evaporation as a quantum channel, the intrinsic information loss $K(\rho,\mathcal{E})$ could quantify how much information about the in-falling state remains inaccessible through measurements of the Hawking radiation. It would therefore be interesting to evaluate $K(\rho,\mathcal{E})$ in semiclassical evaporation models and, ultimately, in more realistic descriptions of evaporating black holes. Such studies could provide a direct characterization of the time-dependent flow and recoverability of black hole information.

It would also be interesting to apply our framework in a thermodynamic setting, where a quantum system interacts with a thermal reservoir at inverse temperature $\beta$. In such a context, Landauer's principle takes the form of an inequality $\beta \Delta Q\ge\Delta S$, where $\Delta S$ is the change in entropy of the system and $\Delta Q$ is the heat dissipated to the reservoir. Because the intrinsic information loss $K(\rho,\E)$ is independent of any observer's subjective choices regarding preparation ensembles and measurement strategies, it characterizes an objective, physical feature of the state-channel dynamics. Therefore, modeling the system-reservoir interaction as a quantum channel $\E$ and exploring it in terms of $K(\rho,\E)$ presents a promising avenue to investigate how the fundamental limits of information retrieval relate to physical energy dissipation. Future work in this direction may therefore uncover novel, channel-theoretic corrections to Landauer's principle for finite-size reservoirs~\cite{Reeb_2014}, potentially bounding the thermodynamic cost of information loss induced by specific open-system dynamics.\\~\\

\emph{Acknowledgments.}---J.F. is supported by Hainan University startup fund for the project ``Spacetime from Quantum Information", and also by the Hainan Provincial Natural Science Foundation of China under Grant No.~126MS0010. W.Z.G is supported by the Seventh Young Faculty Development Program of Huazhong University of Science and Technology, and also by the Hubei Provincial Natural Science Foundation of China under grant No.~2025AFB557.


\bibliography{references}

\clearpage
\newpage


\maketitle
\onecolumngrid
\vspace{1cm}

\begin{center}\large \textbf{Quantum information loss} \\
\textbf{--- Supplemental Material ---}\\
\end{center}

\appendix

\vspace{1em}

In this Supplemental Material we provide proofs of Theorems~1-4 from main text. The proofs were found with the help of AI tools, and were carefully checked by the authors.

\section{Proof of Theorem~1} \label{app:A}

\begin{theoremone} 
Let $\{(p_i,\rho_i)\}$ be an ensemble of states of $A$, and let $\mathcal{E}:\mathcal{L}(\mathcal{H}_A)\to \mathcal{L}(\mathcal{H}_B)$ be a quantum channel. Then there exists a POVM $\{N_j\}$ on system $B$ such that $K=0$ if and only if $\mathcal{E}(\rho_i)\mathcal{E}(\rho_k)=0$ for all $i\neq k$. 
\end{theoremone}

\begin{proof}
\noindent $(\Rightarrow)$
Suppose there exists a POVM $\{N_j\}_{j\in J}$ on $B$ such that $K=0$. Let $f:I\to J$ be the Markov kernel induced by this PEM scenario, with $Q(j|i)=\Tr(\E(\rho_i)N_j)$. By Theorem~\ref{ECXN57}, there exists a perfect retrodiction map $g:J\to I$ such that $g\circ f=\id_I$.

For each $k\in I$, let $F_k=\sum_{j:g(j)=k}N_j$. Since the fibers of $g$ partition $J$, the collection $\{F_k\}_{k\in I}$ is a POVM on $B$. Moreover, by the definition of composition, $\Tr[\E(\rho_i)F_k]=(g\circ f)(k|i)$. Hence
\[
\Tr[\E(\rho_i)F_k]=\delta_{ki}
\]
for all $i,k\in I$.

Recall that if $A,B\geq 0$ and $\Tr(AB)=0$, then $AB=BA=0$. Indeed, $A^{1/2}BA^{1/2}$ is positive with trace zero and hence vanishes. Since $A^{1/2}BA^{1/2}=(B^{1/2}A^{1/2})^\dagger(B^{1/2}A^{1/2})$, it follows that $B^{1/2}A^{1/2}=0$, and hence $BA=AB=0$.

Now fix $i\neq k$. Since $\Tr[\E(\rho_i)F_k]=0$, it follows that $\E(\rho_i)F_k=0$. Moreover, $\Tr[\E(\rho_k)F_k]=1$, and therefore $\Tr[\E(\rho_k)(\mathds{1}_B-F_k)]=0$. It follows that $(\mathds{1}_B-F_k)\E(\rho_k)=0$, so $\E(\rho_k)=F_k\E(\rho_k)$. Thus
\[
\E(\rho_i)\E(\rho_k)=\E(\rho_i)F_k\E(\rho_k)=0.
\]
Hence $\E(\rho_i)\E(\rho_k)=0$ for all $i\neq k$.

\smallskip

\noindent $(\Leftarrow)$
Conversely, suppose that $\E(\rho_i)\E(\rho_k)=0$ for all $i\neq k$. Since the operators $\E(\rho_i)$ are positive, their supports are pairwise orthogonal. Let $\Pi_i$ be the orthogonal projector onto $\operatorname{supp}(\E(\rho_i))$, and let $N_{\mathrm{rest}}=\mathds{1}_B-\sum_{i\in I}\Pi_i$. Then $\{\Pi_i\}_{i\in I}\cup\{N_{\mathrm{rest}}\}$ is a PVM, hence a POVM, on $B$.

Take $J=I\sqcup\{\mathrm{rest}\}$, and let $f:I\to J$ be the Markov kernel induced by this measurement. Then $Q(k|i)=\Tr[\E(\rho_i)\Pi_k]=\delta_{ki}$ and $Q(\mathrm{rest}|i)=0$. Fix $i_0\in I$ and define $g:J\to I$ by $g(i)=i$ for $i\in I$ and $g(\mathrm{rest})=i_0$. Then $g\circ f=\id_I$, so $g$ is a perfect retrodiction map. By Theorem~\ref{ECXN57}, the corresponding PEM scenario satisfies $K=0$.
\end{proof}

\section{Proof of Theorem~2}

\begin{theoremtwo} 
Let $(\{(p_i,\rho_i)\}_{i\in I},\mathcal{E},\{N_j\}_{j\in J})$ be a PEM scenario. Then $K=0$ if and only if there exists a function $g:J\to I$ such that $g\circ f=\id_I$. Such a function is called a perfect retrodiction map.
\end{theoremtwo}

\begin{proof}
Let $f:I\to J$ be the Markov kernel induced by the PEM scenario, with $Q(j|i)=\Tr(\E(\rho_i)N_j)$, and let $Q(j)=\sum_i p_iQ(j|i)$. For every $j$ with $Q(j)>0$, let $P(i|j)=p_iQ(j|i)/Q(j)$ be the posterior probability that Alice prepared $\rho_i$ given that Bob obtained outcome $j$. Then
\[
K=\sum_{j:Q(j)>0}Q(j)\left(-\sum_{i\in I}P(i|j)\log P(i|j)\right).
\]
Since every term in this sum is nonnegative, $K=0$ if and only if, for every $j$ with $Q(j)>0$, the distribution $\{P(i|j)\}_{i\in I}$ is concentrated at a single preparation label.

Suppose first that $K=0$. For every $j$ with $Q(j)>0$, let $g(j)$ be the unique label such that $P(g(j)|j)=1$, and define $g(j)$ arbitrarily when $Q(j)=0$. If $Q(j|i)>0$, then $Q(j)\geq p_iQ(j|i)>0$. Since $p_i>0$, we also have $P(i|j)>0$, and hence $g(j)=i$. It follows that, for all $i,k\in I$,
\[
(g\circ f)(k|i)=\sum_{j:g(j)=k}Q(j|i)=\delta_{ki}.
\]
Thus $g\circ f=\id_I$, so $g$ is a perfect retrodiction map.

Conversely, suppose there exists a perfect retrodiction map $g:J\to I$, so that $g\circ f=\id_I$. If $Q(j|i)>0$, then $g(j)=i$; otherwise, $(g\circ f)(g(j)|i)\geq Q(j|i)>0$, contradicting $g\circ f=\id_I$.

Now fix $j$ with $Q(j)>0$. Since $Q(j)=\sum_i p_iQ(j|i)>0$, there exists an $i$ such that $Q(j|i)>0$, and hence $i=g(j)$. Moreover, $Q(j|k)=0$ for every $k\neq g(j)$. Therefore,
\[
P(i|j)=\frac{p_iQ(j|i)}{Q(j)}=\delta_{i,g(j)}.
\]
Thus, for every occurring outcome $j$, the posterior distribution is concentrated entirely at the single label $g(j)$. Its entropy is therefore zero, and hence $K=0$.
\end{proof}


\section{Proof of Theorem 3}
\label{sec:supp-proof-thm3}

\begin{theoremthree}
Suppose $\{E_{\alpha}\}$ is a collection of Kraus operators for the channel $\E$, and let $P$ be the orthogonal projection operator onto $\mathcal{H}_{\text{code}}$. Then the following statements are equivalent.
\begin{enumerate}
\item 
$\E$ is universally pristine with respect to $\mathcal{H}_{\emph{code}}$.
\item \label{THRMX2711}
$K(\rho,\E)=0$ for every state $\rho$ supported on $\mathcal{H}_{\emph{code}}$.
\item 
There exists a positive semi-definite
matrix $C=(c_{\alpha\beta})$ such that $P E_{\alpha}^{\dagger}E_{\beta}P=c_{\alpha\beta}P$ for all Kraus indices $\alpha,\beta$.
\item 
There exists a quantum channel $\mathcal R:\mathcal L(\mathcal{H}_B)
  \to \mathcal L(\mathcal H_{\emph{code}})$ such that $(\mathcal{R}\circ\E)(\rho)=\rho$ for every state $\rho$
supported on $\mathcal{H}_{\emph{code}}$.
\end{enumerate}
\end{theoremthree}

\begin{proof}
Throughout this proof, we write
$\psi:=|\psi\rangle\langle\psi|$ for a pure state.  We shall repeatedly use
that, for positive semidefinite operators $X,Y\geq0$,
\begin{equation}
 \Tr(XY)=0 \quad\Longleftrightarrow\quad XY=0.
 \label{eq:supp-positive-orthogonality}
\end{equation}
Indeed, $X^{1/2}YX^{1/2}\geq0$ has zero trace if and only if it vanishes,
which is equivalent to $Y^{1/2}X^{1/2}=0$.  For density operators, these
conditions are also equivalent to orthogonality of their supports and to
unit trace distance.  If $\dim\mathcal H_{\mathrm{code}}=1$, all four
conditions in Theorem~3 are immediate; below we therefore assume
$\dim\mathcal H_{\mathrm{code}}\geq2$.

We prove the equivalence of conditions (i)--(iii) by the cycle
$(\mathrm{i})\Rightarrow(\mathrm{ii})\Rightarrow(\mathrm{iii})
\Rightarrow(\mathrm{i})$.

\paragraph*{\boldmath$(\mathrm{i})\Rightarrow(\mathrm{ii})$.}
Let $\rho$ be supported on $\mathcal H_{\mathrm{code}}$, and choose a
spectral decomposition
\begin{equation}
 \rho=\sum_{a=1}^{r}\lambda_a |a\rangle\langle a|,
 \qquad \lambda_a>0.\nonumber
\end{equation}
The eigenvectors are mutually orthogonal.  Universal pristineness therefore
gives
\begin{equation}
 \E(|a\rangle\langle a|)\,
 \E(|b\rangle\langle b|)=0,
 \qquad a\neq b.\nonumber
\end{equation}
By Theorem~1, the spectral ensemble admits a POVM with zero
information loss.  Since $K(\rho,\E)$ is minimized over all pure-state
ensemble decompositions of $\rho$ and all output POVMs, it follows that
$K(\rho,\E)=0$.

\paragraph*{\boldmath$(\mathrm{ii})\Rightarrow(\mathrm{iii})$.}
We first show that condition (ii) forces universal pristineness.  Let
$|\psi\rangle,|\phi\rangle\in\mathcal H_{\mathrm{code}}$ be arbitrary
orthogonal unit vectors, and choose $t\in(0,1)$ with $t\neq\tfrac12$.  The
state
\begin{equation}
 \rho_t=t\psi+(1-t)\phi
\nonumber
\end{equation}
has rank two and a nondegenerate spectrum.  By condition (ii),
$K(\rho_t,\E)=0$.  By the definition of $K$ as a minimum, there is a
pure-state ensemble decomposition
\begin{equation}
 \rho_t=\sum_a q_a\chi_a,
 \qquad \chi_a:=|\chi_a\rangle\langle\chi_a|,
 \qquad q_a>0,
 \label{eq:supp-zero-loss-decomp}
\end{equation}
and an output POVM for which the information loss vanishes.  Theorem~1 then
implies
\begin{equation}
 \E(\chi_a)\E(\chi_b)=0,
 \qquad a\neq b.
 \nonumber
\end{equation}
Thus $D(\E(\chi_a),\E(\chi_b))=1$, where $D$ denotes trace distance.
Contractivity of trace distance under $\E$ yields
\begin{equation}
 1=D(\E(\chi_a),\E(\chi_b))
 \leq D(\chi_a,\chi_b)\leq1,\nonumber
\end{equation}
so the vectors $|\chi_a\rangle$ are mutually orthogonal.  Moreover, every
$|\chi_a\rangle$ belongs to $\operatorname{supp}\rho_t
=\operatorname{span}\{|\psi\rangle,|\phi\rangle\}$.  Consequently,
Eq.~\eqref{eq:supp-zero-loss-decomp} is an orthogonal decomposition of a
rank-two state.  Since $\rho_t$ is nondegenerate, this decomposition must be
its spectral decomposition, up to phases and relabeling.  Therefore,
\begin{equation}
 \E(\psi)\E(\phi)=0.\nonumber
\end{equation}
Because the orthogonal pair $|\psi\rangle,|\phi\rangle$ was arbitrary,
$\E$ is universally pristine on $\mathcal H_{\mathrm{code}}$.

Now write the channel as
\begin{equation}
 \E(X)=\sum_{\alpha}E_{\alpha}XE_{\alpha}^{\dagger}.\nonumber
\end{equation}
For arbitrary pure states $\psi$ and $\phi$,
\begin{align}
 \Tr\!\left[\E(\psi)\E(\phi)\right]
 &=\sum_{\alpha,\beta}
   \Tr\!\left(E_{\alpha}\psi E_{\alpha}^{\dagger}
              E_{\beta}\phi E_{\beta}^{\dagger}\right) \\
 &=\sum_{\alpha,\beta}
   \left|\langle\psi|E_{\alpha}^{\dagger}E_{\beta}|\phi\rangle\right|^2.
 \label{eq:supp-kraus-overlap}
\end{align}
For $|\psi\rangle\perp|\phi\rangle$ the left-hand side vanishes, and hence
\begin{equation}
 \langle\psi|E_{\alpha}^{\dagger}E_{\beta}|\phi\rangle=0
 \label{eq:supp-offdiag-zero}
\end{equation}
for every $\alpha,\beta$.  Define
\begin{equation}
 M_{\alpha\beta}:=P E_{\alpha}^{\dagger}E_{\beta}P.\nonumber
\end{equation}
In any orthonormal basis $\{|a\rangle\}$ of the code subspace,
Eq.~\eqref{eq:supp-offdiag-zero} shows that $M_{\alpha\beta}$ is diagonal.
For $a\neq b$, the orthogonal vectors
\begin{equation}
 |+\rangle=\frac{|a\rangle+|b\rangle}{\sqrt2},
 \qquad
 |-\rangle=\frac{|a\rangle-|b\rangle}{\sqrt2}
 \label{eq:supp-plus-minus}
\end{equation}
satisfy
\begin{equation}
 0=\langle+|M_{\alpha\beta}|-\rangle
 =\frac12\left(
 \langle a|M_{\alpha\beta}|a\rangle
 -\langle b|M_{\alpha\beta}|b\rangle\right).
 \nonumber
\end{equation}
Thus all diagonal entries are equal, so there are numbers $c_{\alpha\beta}$
such that
\begin{equation}
 P E_{\alpha}^{\dagger}E_{\beta}P=c_{\alpha\beta}P.
 \label{eq:supp-KL}
\end{equation}
Finally, $C=(c_{\alpha\beta})$ is positive semidefinite.  Indeed, for any
complex vector $z=(z_{\alpha})$ and any unit vector
$|\psi\rangle\in\mathcal H_{\mathrm{code}}$,
\begin{equation}
 z^{\dagger}Cz
 =\left\langle\psi\left|
 \left(\sum_{\alpha}z_{\alpha}E_{\alpha}\right)^{\dagger}
 \left(\sum_{\beta}z_{\beta}E_{\beta}\right)
 \right|\psi\right\rangle\geq0.
 \nonumber
\end{equation}
This proves condition (iii).

\paragraph*{\boldmath$(\mathrm{iii})\Rightarrow(\mathrm{i})$.}
Assume Eq.~\eqref{eq:supp-KL}.  For orthogonal code states
$|\psi\rangle\perp|\phi\rangle$, Eq.~\eqref{eq:supp-kraus-overlap} gives
\begin{equation}
 \Tr\!\left[\E(\psi)\E(\phi)\right]
 =\sum_{\alpha,\beta}|c_{\alpha\beta}|^2
  |\langle\psi|\phi\rangle|^2=0.
\nonumber
\end{equation}
Equation~\eqref{eq:supp-positive-orthogonality} then implies
$\E(\psi)\E(\phi)=0$, proving universal pristineness.

We have therefore established the equivalence of conditions (i)--(iii).
The equivalence of conditions (iii) and (iv) is precisely the
Knill--Laflamme quantum error-correction theorem cited in the Letter, and is
not repeated here.  This completes the proof of Theorem~3.
\end{proof}

\section{Proof of Theorem 4}
\label{sec:supp-proof-thm4}

For the statement of the Theorem~4, let $\rho_A$ be an arbitrary state with spectral decomposition $\rho_A=\sum_{i=1}^{m}p_i\dyad{i}{i}$. We assume that $m>1$ and $p_i>0$ for all $i$, so that $1<m\leq d_A$. Let $C_p=\sum_{i\neq j}\sqrt{p_ip_j}$, and 
\[
\Delta_{\text{HP}}(\rho_A,s)=\min \left\{1-p_{\text{max}},C_p\sqrt{\frac{d_A(d_{H'}^2-1)}{d_{R'}^2d_{H'}^2-1}}\right\}\, ,
\]
where $p_{\text{max}}=\max\{p_i\}_{i=1}^{m}$. We also let $\epsilon(s)$ be the measure of deviation from universal pristineness given by
\[
\epsilon(s)= \sup_{\substack{
        |\psi\rangle,|\phi\rangle\in\mathcal H_A\\
        \langle\psi|\phi\rangle=0
    }}
    \mathbb E_U
    F_{\mathrm r}
    \left(
        \mathcal E_{U,s}(\psi),
        \mathcal E_{U,s}(\phi)
    \right)\, ,
\]
where $F_{\text{r}}(\rho,\sigma)=\norm{\sqrt{\rho}\sqrt{\sigma}}_1$ is the root fidelity and $\mathbb{E}_U$ denotes the Haar mean. Finally, let $f_m(x)=h_2(x)+x\log(m-1)$, where $h_2$ is the binary entropy.

\begin{theoremfour}
The following statements hold.
\begin{enumerate}
\item
$\epsilon(s)\leq 2^{k/2-s}$.
\item
$\Delta_{\text{HP}}(\rho_A,s)\leq C_p\,2^{k/2-s}$.
\item
$\mathbb{E}_U K(\rho_A,\E_{U,s})\leq f_m\big(\Delta_{\text{HP}}(\rho_A,s)\big)$.
\end{enumerate}
In particular, $\E_{U,s}$ is asymptotically universally pristine on $\mathcal{H}_A$, and the intrinsic information loss $K(\rho_A,\E_{U,s})$ approaches $0$ (in Haar mean) for large $s$.
\end{theoremfour}

\begin{proof}
Let
\begin{equation}
 |\Phi\rangle_{HR}
 =\frac{1}{\sqrt{d_H}}\sum_{a=1}^{d_H}|a\rangle_H|a\rangle_R,
 \qquad \Phi_{HR}:=|\Phi\rangle\langle\Phi|_{HR},
\nonumber
\end{equation}
and set
\begin{equation}
 d:=d_A d_H=d_{R'}d_{H'}.
 \nonumber
\end{equation}
For a pure input $\psi$, define
\begin{equation}
 \sigma_{\psi}^{(U)}:=\E_{U,s}(\psi),
 \qquad
 \omega_{\psi}:=\psi_A\otimes\Phi_{HR}.
\nonumber
\end{equation}
Let $S_X$ denote the swap operator between two replicas of subsystem $X$.
The swap trick gives
\begin{equation}
 \Tr\!\left(\sigma_{\psi}^{(U)}\sigma_{\phi}^{(U)}\right)
 =\Tr\!\left[
 (\omega_{\psi}\otimes\omega_{\phi})
 \left(
 U^{\dagger\otimes2}
 (S_{R'}\otimes I_{H'}^{\otimes2})
 U^{\otimes2}\otimes S_R
 \right)
 \right].
 \label{eq:supp-swap-trick}
\end{equation}
The exact second Haar moment is
\begin{equation}
 \int dU\,
 U^{\dagger\otimes2}(S_{R'}\otimes I_{H'}^{\otimes2})U^{\otimes2}
 =a I_{AH}^{\otimes2}+bS_{AH},
 \label{eq:supp-Haar-twirl}
\end{equation}
where
\begin{equation}
 a=\frac{d_{R'}(d_{H'}^2-1)}{d_{R'}^2d_{H'}^2-1},
 \qquad
 b=\frac{d_{H'}(d_{R'}^2-1)}{d_{R'}^2d_{H'}^2-1}.
 \label{eq:supp-Haar-coefficients}
\end{equation}
These coefficients follow immediately by taking the trace of
Eq.~\eqref{eq:supp-Haar-twirl}, with and without an additional swap
$S_{AH}$.

The two contractions with the initial state are
\begin{align}
 \Tr\!\left[(\omega_{\psi}\otimes\omega_{\phi})
 (I_{AH}^{\otimes2}\otimes S_R)\right]
 &=\frac{1}{d_H},
 \label{eq:supp-contraction-I}\\
 \Tr\!\left[(\omega_{\psi}\otimes\omega_{\phi})
 (S_{AH}\otimes S_R)\right]
 &=|\langle\psi|\phi\rangle|^2.
 \label{eq:supp-contraction-S}
\end{align}
Combining Eqs.~\eqref{eq:supp-swap-trick}--\eqref{eq:supp-contraction-S}
yields
\begin{equation}
 \HU\Tr\!\left(\sigma_{\psi}^{(U)}\sigma_{\phi}^{(U)}\right)
 =\frac{a}{d_H}+b|\langle\psi|\phi\rangle|^2.
 \label{eq:supp-average-overlap-general}
\end{equation}
For orthogonal inputs,
\begin{equation}
 \HU\Tr\!\left(\sigma_{\psi}^{(U)}\sigma_{\phi}^{(U)}\right)
 =\frac{d_{R'}(d_{H'}^2-1)}
 {d_H\bigl(d_{R'}^2d_{H'}^2-1\bigr)}.
 \label{eq:supp-average-overlap-orthogonal}
\end{equation}

Before tracing out $H'$, the state generated from a pure input is pure;
therefore
\begin{equation}
 \operatorname{rank}\sigma_{\psi}^{(U)}\leq d_{H'}.
 \nonumber
\end{equation}
Using $\|X\|_1^2\leq\operatorname{rank}(X)\|X\|_2^2$ with
$X=\sqrt{\sigma_{\psi}^{(U)}}\sqrt{\sigma_{\phi}^{(U)}}$, we obtain
\begin{equation}
 F_{\text{r}}\!\left(\sigma_{\psi}^{(U)},\sigma_{\phi}^{(U)}\right)^2
 \leq d_{H'}\Tr\!\left(\sigma_{\psi}^{(U)}
                              \sigma_{\phi}^{(U)}\right)\, ,
\nonumber
\end{equation}
where $F_{\text{r}}(\rho,\sigma)=\norm{\sqrt{\rho}\sqrt{\sigma}}_1$ is the root fidelity. Cauchy--Schwarz, Eq.~\eqref{eq:supp-average-overlap-orthogonal}, and
$d_A d_H=d_{R'}d_{H'}$ then yields
\begin{align}
 \HU F_{\text{r}}\!\left(\sigma_{\psi}^{(U)},\sigma_{\phi}^{(U)}\right)
 &\leq
 \sqrt{d_{H'}\,
 \HU\Tr\!\left(\sigma_{\psi}^{(U)}\sigma_{\phi}^{(U)}\right)} \nonumber \\
 &=\sqrt{\frac{d_A(d_{H'}^2-1)}
 {d_{R'}^2d_{H'}^2-1}}
 \leq\frac{\sqrt{d_A}}{d_{R'}}.
 \label{eq:supp-universal-fidelity-bound}
\end{align}
The bound is independent of the chosen orthogonal pair.  Since
$d_A=2^k$ and $d_{R'}=2^s$, taking the supremum proves
\begin{equation}
 \epsilon(s)\leq2^{k/2-s},
 \label{eq:supp-epsilon-bound}
\end{equation}
which is statement (i).

For the spectral decomposition
\begin{equation}
 \rho_A=\sum_{i=1}^{m}p_i|i\rangle\langle i|,
 \nonumber
\end{equation}
define
\begin{equation}
 q_s:=\sqrt{\frac{d_A(d_{H'}^2-1)}
 {d_{R'}^2d_{H'}^2-1}}.
\nonumber
\end{equation}
By definition,
\begin{equation}
 \Delta_{\mathrm{HP}}(\rho_A,s)
 =\min\{1-p_{\max},C_p q_s\}
 \leq C_p q_s
 \leq C_p 2^{k/2-s},
 \label{eq:supp-Delta-bound}
\end{equation}
which proves statement (ii).

It remains to bound the information loss.  For fixed $U$, let
\begin{equation}
 \sigma_i^{(U)}:=\E_{U,s}(|i\rangle\langle i|)
\nonumber
\end{equation}
and let $p_{\mathrm{err}}^{\star}(U)$ be the minimum average error probability
for discriminating the ensemble $\{p_i,\sigma_i^{(U)}\}_{i=1}^{m}$.
The Barnum--Knill bound implies that the pretty-good measurement
satisfies~\cite{BaKn02}
\begin{equation}
 p_{\mathrm{err}}^{\mathrm{PGM}}(U)
 \leq
 \sum_{i\neq j}\sqrt{p_ip_j}\,
 F_{\mathrm r}\!\left(
 \sigma_i^{(U)},\sigma_j^{(U)}
 \right).
\end{equation}
Since the optimal minimum-error probability cannot exceed the error
probability of the pretty-good measurement, the same bound holds for
$p_{\mathrm{err}}^\star(U)$,
\begin{equation}
 p_{\mathrm{err}}^{\star}(U)
 \leq\sum_{i\neq j}\sqrt{p_i p_j}\,
 F_{\text{r}}\!\left(\sigma_i^{(U)},\sigma_j^{(U)}\right).
 \label{eq:supp-discrimination-bound}
\end{equation}
Here the sum is over ordered pairs, in agreement with the definition
$C_p=\sum_{i\neq j}\sqrt{p_i p_j}$ used in the main text. Averaging
Eq.~\eqref{eq:supp-discrimination-bound} and using
Eq.~\eqref{eq:supp-universal-fidelity-bound} gives
\begin{equation}
 \HU p_{\mathrm{err}}^{\star}(U)\leq C_p q_s.\nonumber
\end{equation}
Bob may also ignore the output and always guess the most probable label, so
$p_{\mathrm{err}}^{\star}(U)\leq1-p_{\max}$ for every $U$.  Consequently,
\begin{equation}
 \HU p_{\mathrm{err}}^{\star}(U)
 \leq\Delta_{\mathrm{HP}}(\rho_A,s).
 \label{eq:supp-mean-error-Delta}
\end{equation}

Let $\widehat I$ be the estimate produced by an optimal minimum-error POVM.
The spectral ensemble and this POVM are admissible choices in the
minimization defining the intrinsic information loss; hence
\begin{equation}
 K(\rho_A,\E_{U,s})\leq H(I|\widehat I).
\end{equation}
Fano's inequality gives
\begin{equation}
 H(I|\widehat I)
 \leq h_2\!\left(p_{\mathrm{err}}^{\star}(U)\right)
 +p_{\mathrm{err}}^{\star}(U)\log(m-1)
 =f_m\!\left(p_{\mathrm{err}}^{\star}(U)\right).\nonumber
\end{equation}
The function $f_m$ is concave and is increasing on
$0\leq x\leq1-1/m$.  Moreover,
$p_{\mathrm{err}}^{\star}(U)\leq1-p_{\max}\leq1-1/m$.  Jensen's inequality,
Eq.~\eqref{eq:supp-mean-error-Delta}, and monotonicity therefore imply
\begin{align}
 \HU K(\rho_A,\E_{U,s})
 &\leq\HU f_m\!\left(p_{\mathrm{err}}^{\star}(U)\right) \nonumber \\
 &\leq f_m\!\left(\HU p_{\mathrm{err}}^{\star}(U)\right) \nonumber \\
 &\leq f_m\!\left(\Delta_{\mathrm{HP}}(\rho_A,s)\right).
 \label{eq:supp-final-K-bound}
\end{align}
This is statement (iii) and completes the proof of Theorem~4.
\end{proof}

\end{document}